\documentclass[pra, aps, twocolumn, preprintnumbers,showpacs, floatfix,superscriptaddress]{revtex4-1}

\usepackage{hyperref}
\usepackage{latexsym}
\usepackage{mathrsfs}
\usepackage{amssymb,amsmath,bm}
\usepackage{amsfonts}
\usepackage{ifthen}
\usepackage{graphicx}
\usepackage{bm}
\usepackage{color}
\usepackage[mathlines,displaymath]{lineno}
\linenumbers

\newtheorem{theorem}{Theorem}
\newenvironment{proof}{\noindent\textit{Proof}: }{\hfill$\blacksquare$\vskip 0.5\baselineskip}

\begin{document}

\title{On the impossibility of non-static quantum bit commitment between two parties}

\author{Qin Li}
\email[]{liqin805@163.com}
\affiliation{College of Information Engineering,
Xiangtan University, Xiangtan 411105, China}
\affiliation{Department of Computer Science, Sun Yat-sen University,
Guangzhou 510006, China}

\author{Chengqing Li}
\affiliation{College of Information Engineering,
Xiangtan University, Xiangtan 411105, China}

\author{Dong-Yang Long}
\affiliation{Department of Computer Science, Sun Yat-sen University,
Guangzhou 510006, China}

\author{W. H. Chan}
\affiliation{Department of Mathematics,
Hong Kong Baptist University, Kowloon, Hong Kong}

\author{Chun-Hui Wu}
\affiliation{Department of Computer Science, Guangdong University of Finance, Guangzhou 510521, China}

\date{\today}

\begin{abstract}
Recently, Choi \emph{et al}. proposed an assumption on Mayers-Lo-Chau (MLC) no-go theorem that the state of the entire
quantum system is invariable to both participants before the unveiling phase. This means that the theorem
is only applicable to static quantum bit commitment (QBC). This paper find that the assumption is unnecessary and
the MLC no-go theorem can be applied to not only static QBC, but also non-static one. A non-static QBC protocol
proposed by Choi \emph{et al.} is briefly reviewed and analyzed to work as a supporting example. In addition, a novel way to prove the impossibility of the two kinds of QBC is given.
\end{abstract}


\maketitle

\section{Introduction}

Bit commitment allows a sender (Alice) to commit a bit $b\in\{0,1\}$ to a
receiver (Bob) in the following way: 1) Alice
can not change the value of the committed bit after the commitment
phase (binding property); 2) Bob can not obtain the value of the
committed bit before the unveiling phase (concealing property). Bit commitment
is an important cryptographic primitive and can be used as a
building block for some other cryptographic protocols, such as coin
flipping \cite{Nayak:Flip:PRA03}, oblivious transfer \cite{Crepeau:Oblivious:JMO94}, zero-knowledge proof
\cite{Crepeau:proofs:JCSS88}, and multiparty computation \cite{Crepeau:multiparty:LNCS95}.

A secure bit commitment protocol should satisfy the binding property and the concealing property at the same time.
However, unconditionally secure classical bit commitment protocols do not exist. There are only some unconditionally binding and
computationally concealing bit commitment protocols \cite{Naor:commitment:JC94} or
unconditionally concealing and computationally binding bit commitment ones
\cite{Naor:ZeroKnowledge:JC98}. Since unconditionally secure quantum key distribution
protocols were proposed in \cite{Bennett:CoinTossing:ICCSSP84,Ekert:Bell:PRL91,Bennett:nonorthogonal:PRL92}, some
quantum bit commitment (QBC) protocols have been proposed with the hope that QBC can
provide unconditional security \cite{Brassard:CoinToss:LNCS91,Brassard:commitment:ASFCS93,Ardehali:EPR:arXiv96}.
The most famous one is the bit commitment protocol proposed by Brassard \emph{et al}.
in \cite{Brassard:commitment:ASFCS93}, which was claimed to be unconditionally secure.
Unfortunately, the protocol was showed to be insecure afterwards
\cite{Mayers:Trouble:arXiv96}. Furthermore, Mayers, Lo and Chau proved that general secure QBC protocols are impossible
\cite{Mayers:Commitment:PRL97,Lo:CommitPossi:PRL97}, which is called MLC no-go theorem.

Although discovery of the MLC no-go theorem depressed much study on QBC protocols, researchers
try to design secure QBC by adopting certain restrictions or weakening some security
requirements. For instance, Kent proposed two bit commitment protocols based on special relativity theory
\cite{Kent:Commit:PRL99,Kent:ClassicalCommitment:JC05}; Damgard \emph{et al}. designed a secure QBC
protocol in a bounded quantum-storage model \cite{Damgard:QuantumModel:ASFCS05}; Hardy and Kent
gave a secure cheat sensitive QBC protocol ensuring that if
either a committer or a committee cheat, the other can detect it with a nonzero probability \cite{Kent:SensiCommit:PRL04}. Besides these, secure QBC protocols are implemented in a noisy-storage model under the assumption that the dishonest party can not access large-scale reliable quantum storage \cite{Wehner:noisystorage:PRA10}.

Recently, Choi \emph{et al}. proposed a secure non-static QBC protocol with
help of a trusted third party (TTP) \cite{Choi:NonStatic:arXiv09} and
pointed out that the MLC no-go theorem is based on an assumption that the whole quantum state is static
before Alice reveals the committed bit. That is to say, the MLC no-go theorem was thought to be adapted to static QBC only.
D'Ariano \emph{et al}. also hold this opinion and gave another
strengthened and explicit proof involving impossibility of some non-static QBC protocols \cite{DAriano:bitcommitment:PRA07}.
However, we find that
the assumption given by Choi \emph{et al}. in \cite{Choi:NonStatic:arXiv09} is unnecessary and
non-static QBC is also impossible just due to the MLC no-go
theorem. Although Choi \emph{et al}. proposed a secure non-static
QBC protocol by adopting a TTP \cite{Choi:NonStatic:arXiv09}, the
protocol is still different from a general two-party QBC protocol
and is somewhat similar to a quantum secret sharing protocol. Interestingly, the non-static
QBC protocol without the TTP can just serve as an example to show that the MLC no-go theorem can be applied to non-static QBC also. In addition, we
prove the impossibility of the two kinds of QBC in a different way: prove any binding QBC protocols is not concealing, while the related proofs proposed  in \cite{Mayers:Commitment:PRL97,Lo:CommitPossi:PRL97,Choi:NonStatic:arXiv09,DAriano:bitcommitment:PRA07} show any concealing protocols are not binding.

The rest of this paper is organized as follows. In the next section, it will be shown that the assumption
of the MLC no-go theorem given by Choi \emph{et al}. is unnecessary and the MLC no-go theorem can be adapted to
both static QBC and non-static QBC. The non-static QBC protocol proposed by Choi \emph{et al}. is reviewed and analyzed in Sec.~\ref{sec:review}.
Then the impossibility of QBC is proved in Sec.~\ref{sec:proofimpossibility} in a different way. The last section concludes the paper.

\section{Applicability of the MLC no-go theorem to non-static QBC}
\label{sec:theorem}

In \cite{Mayers:Commitment:PRL97,Lo:CommitPossi:PRL97}, the MLC no-go theorem was proved in the following
basic idea. Suppose the initial states of Alice and Bob are $|b\rangle_A$ ($b\in\{0, 1\}$) and
$|\varphi\rangle_B$, respectively, and let $U_{AB}$ denote all the
algorithms that Alice and Bob may implement. Then the final quantum
state shared by Alice and Bob is
$|\phi_b\rangle_{AB}=U_{AB}(|b\rangle_A\otimes|\varphi\rangle_B)$.
If a QBC protocol is perfectly concealing, namely
\begin{equation*}
\rho_0^B=Tr_A(|\phi_0\rangle_{AB}\langle\phi_0|)=Tr_A(|\phi_1\rangle_{AB}\langle\phi_1|)=\rho_1^B,
\end{equation*}
then there exists a local unitary transformation $S_A$
satisfying
\begin{equation*}
(S_A \otimes
I)U_{AB}(|b\rangle_A\otimes|\varphi\rangle_B)=U_{AB}(|1-b\rangle_A\otimes|\varphi\rangle_B)
\label{eq:SAcondition}
\end{equation*}
according to Gisin-Hughston-Jozsa-Wootters theorem given in
\cite{Gisin:Stochastic:HPA89,Hughstona:classification:PLA93}.
So, by postponing measurements and implementing local unitary
operations, Alice can change the value of the committed bit
arbitrarily without being discovered by Bob. If the QBC protocol is
supposed to be unconditionally concealing, similar results can be
derived also.

However, Choi \emph{et al}. observed that the local unitary
operation $S_A$ performed by Alice is related to Bob's initial state
$|\varphi\rangle_B$ \cite{Choi:NonStatic:arXiv09}. If $|\varphi\rangle_B$ is
random and unknown to Alice, she can not find a suitable local
unitary operation to change the committed value.  Thus, a necessary
assumption of the MLC no-go theorem is that the state of quantum
system should be static to both participants. This means the
MLC no-go theorem was considered to be applicable to static QBC
only.

As shown above, the proof of the MLC no-go theorem is based on the following strategy:
a QBC protocol is first supposed unconditionally concealing
and it is then proved that unconditionally binding is impossible.
So, Theorem~\ref{theorem:nonstaticQBC} can be obtained, which means that the assumption of
the MLC no-go theorem suggested by Choi \emph{et al}. is unnecessary.

\begin{theorem}
The MLC no-go theorem is also applied to non-static QBC.
\label{theorem:nonstaticQBC}
\end{theorem}
\begin{proof}Assume $|\varphi\rangle_B$ is random and unknown to Alice. Let
$U_{AB}=\sum_{ijkl}a_{ijkl}|i\rangle_A|j\rangle_B$$ _A\langle
k|_B\langle l|$ and $|\varphi\rangle_B=\sum_m c_m|m\rangle_B$, then
we have
\begin{eqnarray*}
|\phi_b\rangle_{AB} & = & U_{AB}(|b\rangle_A\otimes|\varphi\rangle_B)\nonumber\\
                    & = & \sum_{ijkl}a_{ijkl}|i\rangle_A|j\rangle_B {_A\langle} k|_B\langle l|\left(\sum_m c_m|b\rangle_A|m\rangle_B\right)\\
                    & = & \sum_{ijl}a_{ijbl}c_l|i\rangle_A|j\rangle_B,
\end{eqnarray*}
and
\begin{eqnarray*}
\rho_b^B & = & Tr_A(|\phi_b\rangle_{AB}\langle \phi_b|)\\
         & = & Tr_A\left(\sum_{ijl}a_{ijbl}c_l|i\rangle_A|j\rangle_B\sum_{pqr}a_{pqbr}^*c_r^*{}_A\langle p|_B\langle q|\right)\\
         & = & \sum_{ijlqr}a_{ijbl}a_{iqbr}^*{c_l}{c_r^*}|j\rangle_B\langle q|.
\end{eqnarray*}

Suppose a non-static QBC protocol is perfectly concealing, then
$\rho_0^B$ and $\rho_1^B$ should be identical for any
$|\varphi\rangle_B=\sum_m c_m|m\rangle_B$, i.e.
\begin{eqnarray*}
\rho_0^B & = & \sum_{ijlqr}a_{ij0l}a_{iq0r}^*{c_l}{c_r^*}|j\rangle_B\langle q|\\
         & = & \sum_{ijlqr}a_{ij1l}a_{iq1r}^*{c_l}{c_r^*}|j\rangle_B\langle q|\\
         & = & \rho_1^B.
\end{eqnarray*}

Since the above formula always holds for any $|\varphi\rangle_B$, we have
\begin{equation}
a_{ij0l}a_{iq0r}^*=a_{ij1l}a_{iq1r}^*.
\label{eq:equalityaij}
\end{equation}

Let $S_A=\sum\limits_{xy}s_{xy}|x\rangle_A\langle y|$, where
\begin{equation*}
\left\{ \begin{array}{ll}
s_{xy}  =  0,                             & \text{if } x \ne y,\\
s_{xx}  =  0,                             & \text{if } a_{xq1r}=0 \text{ for any } q \text{ and } r,\\
s_{xx}  =  \frac{a_{xq0r}^*}{a_{xq1r}^*}, & \text{otherwise.}\\
\end{array} \right.
\end{equation*}

Equation~(\ref{eq:equalityaij}) makes $S_A$ always satisfy
\begin{eqnarray*}
(S_A\otimes I_B)|\phi_0\rangle_{AB} & = & \left(\sum_{xy}s_{xy}|x\rangle_A\langle y|\otimes \sum_k|k\rangle_B\langle k|\right)\\
                                    &   & {}{ }\left(\sum_{ijl}a_{ij0l} c_l|i\rangle_A|j\rangle_B\right)\\
                                    & = & \sum_{ijl}a_{ij0l} c_l\sum_x s_{xi}|x\rangle_A|j\rangle_B\\
                                    & = & \sum_{ijl}a_{ij1l} c_l|i\rangle_A|j\rangle_B\\
                                    & = & |\phi_1\rangle_{AB},
\end{eqnarray*}
for any $|\varphi\rangle_B$. Thus the non-static QBC protocol is not binding. If assume it is unconditionally concealing,
similar results can be obtained also.
\end{proof}

\section{Review and analysis of Choi \emph{et al}.'s non-static QBC protocol}
\label{sec:review}

In \cite{Choi:NonStatic:arXiv09}, Choi \emph{et al}. proposed an
unconditionally secure non-static QBC protocol with aid of a TTP.
We briefly reviewed it in the following four phases and then show
its simplified version can serve as an example demonstrating the MLC no-go theorem is
applicable to non-static QBC.

\paragraph{Preparing phase}: First, Alice and TTP share $N$ maximally entangled states in the form $|\psi^-\rangle_{AT}=\frac{|01>-|10>}{\sqrt{2}}$. This kind of entangled state has a special property, namely equation
\begin{equation*}
|\psi^-\rangle_{AT}=(U\otimes U)|\psi^-\rangle_{AT}
\end{equation*}
holds up to the global phase for any unitary transformation $U$. Then, TTP applies random projection measurements represented as
\begin{equation*}
M_i=\{|f_i\rangle_T\langle f_i|,|f_i^\bot\rangle_T\langle f_i^\bot|\}
\label{eq:two}
\end{equation*}
to its qubit of each entangled state $|\psi^-\rangle_{AT}$ for $i=1 \sim N$.
If the measurement outcome of TTP is
$|f_i\rangle(|f_i^\bot\rangle)$, then Alice's measurement result
should be $|\psi_i\rangle_A=|f_i^\bot\rangle(|f_i\rangle)$. But
TTP does not announce $M_i$ now, so Alice can not know the result
$|\psi_i\rangle_A$.

\paragraph{Commitment phase}:
To committee the bit $b$, Alice applies the corresponding operations $P_i\in\{M,N,J,K\}$, where
\begin{eqnarray*}
M = \left (
\begin{array}{cc}
1 & 0 \\
0 & 1
\end{array}\right),
N = \left (
\begin{array}{cc}
0 & -1 \\
1 & 0
\end{array}\right),\\
J = \frac{1}{\sqrt{2}}\left (
\begin{array}{cc}
1 & i \\
1 & -i
\end{array}\right),
K = \frac{1}{\sqrt{2}}\left (
\begin{array}{cc}
1 & i \\
-1 & i
\end{array}\right).
\end{eqnarray*}
If Alice chooses to commit $b=0$, she randomly sends
$M|\psi_i\rangle_A$ or $N|\psi_i\rangle_A$ to Bob. Otherwise, she
sends $J|\psi_i\rangle_A$ or $K|\psi_i\rangle_A$ instead with the same
probability. To guarantee the randomness, Alice introduces an
auxiliary system $A'$ whose initial state is
$|+\rangle_{A'}=\frac{|0\rangle_{A'}+|1\rangle_{A'}}{\sqrt{2}}$.
Then the state of the whole system $A'A$ is
$|+\rangle|\psi_i\rangle$. If $b=0$, Alice applies
$|0\rangle_{A'}\langle 0|\otimes M+|1\rangle_{A'}\langle 1|\otimes
N$ to $A'A$ to obtain
\begin{equation*}
|\varphi_0\rangle_{A'A}=\frac{|0\rangle_{A'}\otimes M|\psi_i\rangle_A+|1\rangle_{A'}\otimes N|\psi_i\rangle_A}{\sqrt{2}}.
\end{equation*}
Otherwise, she implements $|0\rangle_{A'}\langle 0|\otimes
J+|1\rangle_{A'}\langle 1|\otimes K$ on $A'A$ and gets
\begin{equation*}
|\varphi_1\rangle_{A'A}=\frac{|0\rangle_{A'}\otimes J|\psi_i\rangle_A+|1\rangle_{A'}\otimes K|\psi_i\rangle_A}{\sqrt{2}}.
\end{equation*}
Due to the randomness of $|\psi_i\rangle_A$, the resulting state
$|\varphi_b\rangle_{A'A}$ is also different, thus Alice cannot
control the relationship between $|\varphi_0\rangle_{A'A}$ and
$|\varphi_1\rangle_{A'A}$ without knowing the exact state
$|\psi_i\rangle_A$.

\paragraph{Sustaining phase}:
In this phase, both Alice and Bob do nothing.

\paragraph{Revealing phase}: Alice unveils all $P_i$'s, and then TTP opens all $M_i$'s and the corresponding measurement outcomes.
After knowing all the information, Bob measures $P_i^\dagger P_i|\psi_i\rangle$ with $M_i$ and compares
all the measurement results with those TTP announced. If all the measurement outcomes are opposite,
Alice is honest and the committed value has not been changed; otherwise Alice is dishonest.

In \cite{Choi:NonStatic:arXiv09}, Choi \emph{et al}. claimed that the protocol is unconditionally secure.
However, the usage of TTP makes the above non-static QBC protocol do not
correspond to the fact that only two parties is involved in a
general QBC protocol, although TTP plays a little role in offering
quantum sources and is not involved in communication between two
parties directly. In a way, the protocol is more like a quantum secret sharing
protocol. For instance, the cooperation between Bob and TTP can get Aice's committed value while one of them cannot. If the actions
implemented by TTP are replaced by Bob, the protocol will not be
secure. As shown by Choi \emph{et al}. in \cite{Choi:NonStatic:arXiv09}, if the non-static QBC
protocol without a TTP is perfectly concealing, a local unitary
operator
$$S_A = \left(
\begin{array}{ll}
 a & b \\
 c & d \\
\end{array} \right)$$
such that $J=aM+bN$ and $K=cM+dN$, can be used to freely change the committed bit.
Thus it can be seen that Choi \emph{et al}.'s non-static QBC protocol without a TTP can
serve as a specific example to demonstrate that the MLC no-go theorem is
applicable to non-static QBC also.

\section{Proof of impossibility of QBC by another way}
\label{sec:proofimpossibility}

Although non-static QBC between two participants is also impossible due to the MLC no-go theorem,
it provides us another way to prove the impossibility of both non-static and static QBC.

Let us show the case on non-static QBC first. Premise of the proof of the MLC no-go theorem is that the QBC
protocol is supposed to be perfectly concealing,
\begin{equation}
F(\rho_0^B,\rho_1^B)=1,
\label{eq:perfectconcealcondition}
\end{equation}
or unconditionally concealing,
\begin{equation*}
F(\rho_0^B,\rho_1^B)=1-\delta,
\end{equation*}
where $\delta>0$.

For non-static QBC protocols, different
initial states $|\varphi\rangle_B$ may lead to different $\rho_b^B$,
so the value of $F(\rho_0^B,\rho_1^B)$ may vary and
it is difficult to make the concealing property be satisfied. On the other hand,
since $|\varphi\rangle_B$ is totally random and unknown to Alice, it
is difficult for Alice to find an appropriate local unitary operator
$S_A$ such that
\begin{equation}
F((S_A \otimes
I)U_{AB}(|b\rangle_A\otimes|\varphi\rangle_B),U_{AB}(|1-b\rangle\otimes|\varphi\rangle_B))=1
\label{eq:FSAcondition}
\end{equation}
or
\begin{equation*}
F((S_A \otimes
I)U_{AB}(|b\rangle_A\otimes|\varphi\rangle_B),U_{AB}(|1-b\rangle\otimes|\varphi\rangle_B))=1-\delta
\end{equation*}
is satisfied for all $|\varphi\rangle_B$.

Thus, it is better to suppose a non-static QBC protocol is perfectly or
unconditionally binding and then prove it cannot be perfectly or
unconditionally concealing.

Assume a non-static QBC protocol is perfectly binding, i.e., there
does not exist a local unitary operator $S_A$ such that Eq.~(\ref{eq:FSAcondition}) holds for any $|\varphi\rangle_B$. Then there must be
some $|\varphi\rangle_B$ such that
\begin{equation*}
\rho_0^B=Tr_A(|\phi_0\rangle_{AB}\langle \phi_0|)\neq
Tr_A(|\phi_1\rangle_{AB}\langle \phi_1|)=\rho_1^B.
\end{equation*}
Otherwise, the assumption violates the MLC no-go theorem. In other
words, if Eq.~(\ref{eq:perfectconcealcondition}) holds for any $|\varphi\rangle_B$, Alice
can find a local unitary operation $S_A$ to freely change the
committed bit according to the MLC no-go theorem. Thus Bob can choose such $|\varphi\rangle_B$ to
get some information of Alice's committed bit and the non-static QBC
is not perfectly concealing. If a non-static QBC protocol is assumed
to be unconditionally binding, similar conclusions can be made.

This new approach also can be used to prove impossibility of
static QBC. Given a fixed $|\varphi\rangle_B$, if a static QBC
protocol is supposed to be perfectly binding, then there is no local
unitary operator $S_A$ satisfying Eq.~(\ref{eq:FSAcondition}). According to the
Uhlmann's theorem in \cite{Jozsa:MixedQuantum:JMO94}, we can find $|\phi\rangle$, a
purification of $\rho_0^B$, such that
\begin{equation*}
F(\rho_0^B,\rho_1^B)=|\langle \phi|\phi_1\rangle|,
\end{equation*}
where $|\phi_1\rangle=U_{AB}(|1\rangle_A\otimes|\varphi\rangle_B)$
is a purification of $\rho_1^B$. Between two purifications of
$\rho_0^B$, $|\phi\rangle$ and
$|\phi_0\rangle=U_{AB}(|0\rangle_A\otimes|\varphi\rangle_B)$, there
always exist a local unitary operator $S_A$ such that $(S_A\otimes
I)|\phi_0\rangle=|\phi\rangle$. Besides, from the assumption, we
know that there does not exist a local unitary operator $S_A$ such
that $(S_A\otimes I)|\phi_0\rangle=|\phi_1\rangle$. Thus
$|\phi\rangle$ cannot be equal to $|\phi_1\rangle$ and
\begin{equation*}
F(\rho_0^B,\rho_1^B)=|\langle \phi|\phi_1\rangle|\neq 1,
\end{equation*}
which means the static QBC protocol is not perfectly concealing. If
assume a static QBC protocol is unconditionally binding, we can
prove it is not unconditionally concealing employing the similar
method.

\section{Conclusion}

In this paper, we show the assumption given by Choi \emph{et al}. on
the MLC no-go theorem in \cite{Choi:NonStatic:arXiv09}, that the entire quantum state should be static to
both participants before the unveiling phase, is
unnecessary, and the MLC no-go theorem can be applied to both static QBC
and non-static QBC. In addition, a secure non-static QBC protocol proposed by Choi \emph{et al}. in \cite{Choi:NonStatic:arXiv09}
is found more like to a quantum secret sharing protocol, instead of a general two-party QBC protocol. Just
inspired by the non-static QBC, we prove the impossibility of QBC in another way: suppose a QBC protocol is binding first, then
show it is not concealing. Now, we can say that the MLC no-go
theorem lets any two-party QBC protocol satisfying concealing property is not
binding and the novel proof for the impossibility of QBC given by us makes any two-party QBC protocol satisfying binding property is not
concealing. In all, any two-party QBC protocol, no matter static or non-static, is not secure.

\section*{Acknowledgements}

The authors would like to thank Guang-Ping He and
Chun-Yuan Lu for their constructive suggestions on improving this paper.
The work of Qin Li and Chengqing Li was supported by start-up funding of Xiangtan University of grant number 10QDZ39.
The work of W. H. Chan was partially supported by the Faculty Research Grant of Hong Kong
Baptist University under grant number FRG2/08-09/070.

\bibliographystyle{apsrev4-1}
\bibliography{NSQBC}

\begin{thebibliography}{25}%
\makeatletter
\providecommand \@ifxundefined [1]{%
 \@ifx{#1\undefined}
}%
\providecommand \@ifnum [1]{%
 \ifnum #1\expandafter \@firstoftwo
 \else \expandafter \@secondoftwo
 \fi
}%
\providecommand \@ifx [1]{%
 \ifx #1\expandafter \@firstoftwo
 \else \expandafter \@secondoftwo
 \fi
}%
\providecommand \natexlab [1]{#1}%
\providecommand \enquote  [1]{``#1''}%
\providecommand \bibnamefont  [1]{#1}%
\providecommand \bibfnamefont [1]{#1}%
\providecommand \citenamefont [1]{#1}%
\providecommand \href@noop [0]{\@secondoftwo}%
\providecommand \href [0]{\begingroup \@sanitize@url \@href}%
\providecommand \@href[1]{\@@startlink{#1}\@@href}%
\providecommand \@@href[1]{\endgroup#1\@@endlink}%
\providecommand \@sanitize@url [0]{\catcode `\\12\catcode `\$12\catcode
  `\&12\catcode `\#12\catcode `\^12\catcode `\_12\catcode `\%12\relax}%
\providecommand \@@startlink[1]{}%
\providecommand \@@endlink[0]{}%
\providecommand \url  [0]{\begingroup\@sanitize@url \@url }%
\providecommand \@url [1]{\endgroup\@href {#1}{\urlprefix }}%
\providecommand \urlprefix  [0]{URL }%
\providecommand \Eprint [0]{\href }%
\providecommand \doibase [0]{http://dx.doi.org/}%
\providecommand \selectlanguage [0]{\@gobble}%
\providecommand \bibinfo  [0]{\@secondoftwo}%
\providecommand \bibfield  [0]{\@secondoftwo}%
\providecommand \translation [1]{[#1]}%
\providecommand \BibitemOpen [0]{}%
\providecommand \bibitemStop [0]{}%
\providecommand \bibitemNoStop [0]{.\EOS\space}%
\providecommand \EOS [0]{\spacefactor3000\relax}%
\providecommand \BibitemShut  [1]{\csname bibitem#1\endcsname}%
\let\auto@bib@innerbib\@empty
\bibitem [{\citenamefont {Nayak}\ and\ \citenamefont
  {Shor}(2003)}]{Nayak:Flip:PRA03}%
  \BibitemOpen
  \bibfield  {author} {\bibinfo {author} {\bibfnamefont {A.}~\bibnamefont
  {Nayak}}\ and\ \bibinfo {author} {\bibfnamefont {P.}~\bibnamefont {Shor}},\
  }\href@noop {} {\bibfield  {journal} {\bibinfo  {journal} {Physical Review
  A}\ }\textbf {\bibinfo {volume} {67}},\ \bibinfo {pages} {article no. 012304}
  (\bibinfo {year} {2003})}\BibitemShut {NoStop}%
\bibitem [{\citenamefont {Crepeau}(1994)}]{Crepeau:Oblivious:JMO94}%
  \BibitemOpen
  \bibfield  {author} {\bibinfo {author} {\bibfnamefont {C.}~\bibnamefont
  {Crepeau}},\ }\href@noop {} {\bibfield  {journal} {\bibinfo  {journal}
  {Journal of Modern Optics}\ }\textbf {\bibinfo {volume} {41}},\ \bibinfo
  {pages} {2445} (\bibinfo {year} {1994})}\BibitemShut {NoStop}%
\bibitem [{\citenamefont {Brassard}\ \emph {et~al.}(1988)\citenamefont
  {Brassard}, \citenamefont {Chaum},\ and\ \citenamefont
  {Crepeau}}]{Crepeau:proofs:JCSS88}%
  \BibitemOpen
  \bibfield  {author} {\bibinfo {author} {\bibfnamefont {G.}~\bibnamefont
  {Brassard}}, \bibinfo {author} {\bibfnamefont {D.}~\bibnamefont {Chaum}}, \
  and\ \bibinfo {author} {\bibfnamefont {C.}~\bibnamefont {Crepeau}},\
  }\href@noop {} {\bibfield  {journal} {\bibinfo  {journal} {Journal of
  Computer and System Sciences}\ }\textbf {\bibinfo {volume} {37}},\ \bibinfo
  {pages} {156} (\bibinfo {year} {1988})}\BibitemShut {NoStop}%
\bibitem [{\citenamefont {Crepeau}\ \emph {et~al.}(1995)\citenamefont
  {Crepeau}, \citenamefont {van~de Graaf},\ and\ \citenamefont
  {Tapp}}]{Crepeau:multiparty:LNCS95}%
  \BibitemOpen
  \bibfield  {author} {\bibinfo {author} {\bibfnamefont {C.}~\bibnamefont
  {Crepeau}}, \bibinfo {author} {\bibfnamefont {J.}~\bibnamefont {van~de
  Graaf}}, \ and\ \bibinfo {author} {\bibfnamefont {A.}~\bibnamefont {Tapp}},\
  }in\ \href@noop {} {\emph {\bibinfo {booktitle} {Advances In
  Cryptology-Crypto'95}}},\ \bibinfo {series} {Lecture Notes in Computer
  Science}, Vol.\ \bibinfo {volume} {963}\ (\bibinfo  {publisher} {Springer},\
  \bibinfo {year} {1995})\ pp.\ \bibinfo {pages} {110--123}\BibitemShut
  {NoStop}%
\bibitem [{\citenamefont {Naor}(1994)}]{Naor:commitment:JC94}%
  \BibitemOpen
  \bibfield  {author} {\bibinfo {author} {\bibfnamefont {M.}~\bibnamefont
  {Naor}},\ }\href@noop {} {\bibfield  {journal} {\bibinfo  {journal} {Journal
  of Cryptology}\ }\textbf {\bibinfo {volume} {4}},\ \bibinfo {pages} {151}
  (\bibinfo {year} {1994})}\BibitemShut {NoStop}%
\bibitem [{\citenamefont {Naor}\ \emph {et~al.}(1998)\citenamefont {Naor},
  \citenamefont {Ostrovsky}, \citenamefont {Venkatesan},\ and\ \citenamefont
  {Yung}}]{Naor:ZeroKnowledge:JC98}%
  \BibitemOpen
  \bibfield  {author} {\bibinfo {author} {\bibfnamefont {M.}~\bibnamefont
  {Naor}}, \bibinfo {author} {\bibfnamefont {R.}~\bibnamefont {Ostrovsky}},
  \bibinfo {author} {\bibfnamefont {R.}~\bibnamefont {Venkatesan}}, \ and\
  \bibinfo {author} {\bibfnamefont {M.}~\bibnamefont {Yung}},\ }\href@noop {}
  {\bibfield  {journal} {\bibinfo  {journal} {Journal of Cryptology}\ }\textbf
  {\bibinfo {volume} {11}},\ \bibinfo {pages} {87} (\bibinfo {year}
  {1998})}\BibitemShut {NoStop}%
\bibitem [{\citenamefont {Bennett}\ and\ \citenamefont
  {Brassard}(1984)}]{Bennett:CoinTossing:ICCSSP84}%
  \BibitemOpen
  \bibfield  {author} {\bibinfo {author} {\bibfnamefont {C.~H.}\ \bibnamefont
  {Bennett}}\ and\ \bibinfo {author} {\bibfnamefont {G.}~\bibnamefont
  {Brassard}},\ }in\ \href@noop {} {\emph {\bibinfo {booktitle} {Proceedings of
  the IEEE International Conference on Computers Systems and Signal
  Processing}}}\ (\bibinfo {year} {1984})\ pp.\ \bibinfo {pages}
  {175--179}\BibitemShut {NoStop}%
\bibitem [{\citenamefont {Ekert}(1991)}]{Ekert:Bell:PRL91}%
  \BibitemOpen
  \bibfield  {author} {\bibinfo {author} {\bibfnamefont {A.~K.}\ \bibnamefont
  {Ekert}},\ }\href@noop {} {\bibfield  {journal} {\bibinfo  {journal}
  {Physical Review Letters}\ }\textbf {\bibinfo {volume} {67}},\ \bibinfo
  {pages} {661} (\bibinfo {year} {1991})}\BibitemShut {NoStop}%
\bibitem [{\citenamefont {Bennett}(1992)}]{Bennett:nonorthogonal:PRL92}%
  \BibitemOpen
  \bibfield  {author} {\bibinfo {author} {\bibfnamefont {C.~H.}\ \bibnamefont
  {Bennett}},\ }\href@noop {} {\bibfield  {journal} {\bibinfo  {journal}
  {Physical Review Letters}\ }\textbf {\bibinfo {volume} {68}},\ \bibinfo
  {pages} {3121} (\bibinfo {year} {1992})}\BibitemShut {NoStop}%
\bibitem [{\citenamefont {Brassard}\ and\ \citenamefont
  {Crepeau}(1991)}]{Brassard:CoinToss:LNCS91}%
  \BibitemOpen
  \bibfield  {author} {\bibinfo {author} {\bibfnamefont {G.}~\bibnamefont
  {Brassard}}\ and\ \bibinfo {author} {\bibfnamefont {C.}~\bibnamefont
  {Crepeau}},\ }in\ \href@noop {} {\emph {\bibinfo {booktitle} {Advances In
  Cryptology-Crypto'90}}},\ \bibinfo {series} {Lecture Notes in Computer
  Science}, Vol.\ \bibinfo {volume} {537}\ (\bibinfo  {publisher} {Springer},\
  \bibinfo {year} {1991})\ pp.\ \bibinfo {pages} {49--61}\BibitemShut {NoStop}%
\bibitem [{\citenamefont {Brassard}\ \emph {et~al.}(1993)\citenamefont
  {Brassard}, \citenamefont {Crepeau}, \citenamefont {Jozsa},\ and\
  \citenamefont {Langlois}}]{Brassard:commitment:ASFCS93}%
  \BibitemOpen
  \bibfield  {author} {\bibinfo {author} {\bibfnamefont {G.}~\bibnamefont
  {Brassard}}, \bibinfo {author} {\bibfnamefont {C.}~\bibnamefont {Crepeau}},
  \bibinfo {author} {\bibfnamefont {R.}~\bibnamefont {Jozsa}}, \ and\ \bibinfo
  {author} {\bibfnamefont {D.}~\bibnamefont {Langlois}},\ }in\ \href@noop {}
  {\emph {\bibinfo {booktitle} {Proceedings of the 34th Annual Symposium on
  Foundations of Computer Science}}}\ (\bibinfo {year} {1993})\ pp.\ \bibinfo
  {pages} {362--371}\BibitemShut {NoStop}%
\bibitem [{\citenamefont {Ardehali}(1996)}]{Ardehali:EPR:arXiv96}%
  \BibitemOpen
  \bibfield  {author} {\bibinfo {author} {\bibfnamefont {M.}~\bibnamefont
  {Ardehali}},\ }\href@noop {} {\enquote {\bibinfo {title} {A quantum bit
  commitment protocol based on {EPR} states},}\ } (\bibinfo {year} {1996}),\
  \bibinfo {note} {arXiv:quant-ph/9505019v5}\BibitemShut {NoStop}%
\bibitem [{\citenamefont {Mayers}(1996)}]{Mayers:Trouble:arXiv96}%
  \BibitemOpen
  \bibfield  {author} {\bibinfo {author} {\bibfnamefont {D.}~\bibnamefont
  {Mayers}},\ }\href@noop {} {\enquote {\bibinfo {title} {The trouble with
  quantum bit commitment},}\ } (\bibinfo {year} {1996}),\ \bibinfo {note}
  {arXiv:quant-ph/9603015v3}\BibitemShut {NoStop}%
\bibitem [{\citenamefont {Mayers}(1997)}]{Mayers:Commitment:PRL97}%
  \BibitemOpen
  \bibfield  {author} {\bibinfo {author} {\bibfnamefont {D.}~\bibnamefont
  {Mayers}},\ }\href@noop {} {\bibfield  {journal} {\bibinfo  {journal}
  {Physical Review Letters}\ }\textbf {\bibinfo {volume} {78}},\ \bibinfo
  {pages} {3414} (\bibinfo {year} {1997})}\BibitemShut {NoStop}%
\bibitem [{\citenamefont {Lo}\ and\ \citenamefont
  {Chau}(1997)}]{Lo:CommitPossi:PRL97}%
  \BibitemOpen
  \bibfield  {author} {\bibinfo {author} {\bibfnamefont {H.-K.}\ \bibnamefont
  {Lo}}\ and\ \bibinfo {author} {\bibfnamefont {H.~F.}\ \bibnamefont {Chau}},\
  }\href@noop {} {\bibfield  {journal} {\bibinfo  {journal} {Physical Review
  Letters}\ }\textbf {\bibinfo {volume} {78}},\ \bibinfo {pages} {3410}
  (\bibinfo {year} {1997})}\BibitemShut {NoStop}%
\bibitem [{\citenamefont {Kent}(1999)}]{Kent:Commit:PRL99}%
  \BibitemOpen
  \bibfield  {author} {\bibinfo {author} {\bibfnamefont {A.}~\bibnamefont
  {Kent}},\ }\href@noop {} {\bibfield  {journal} {\bibinfo  {journal} {Physical
  Review Letters}\ }\textbf {\bibinfo {volume} {83}},\ \bibinfo {pages} {1447}
  (\bibinfo {year} {1999})}\BibitemShut {NoStop}%
\bibitem [{\citenamefont {Kent}(2005)}]{Kent:ClassicalCommitment:JC05}%
  \BibitemOpen
  \bibfield  {author} {\bibinfo {author} {\bibfnamefont {A.}~\bibnamefont
  {Kent}},\ }\href@noop {} {\bibfield  {journal} {\bibinfo  {journal} {Journal
  of Cryptology}\ }\textbf {\bibinfo {volume} {18}},\ \bibinfo {pages} {313}
  (\bibinfo {year} {2005})}\BibitemShut {NoStop}%
\bibitem [{\citenamefont {Damgard}\ \emph {et~al.}(2005)\citenamefont
  {Damgard}, \citenamefont {Fehr}, \citenamefont {Salvail},\ and\ \citenamefont
  {Schaffner}}]{Damgard:QuantumModel:ASFCS05}%
  \BibitemOpen
  \bibfield  {author} {\bibinfo {author} {\bibfnamefont {I.~B.}\ \bibnamefont
  {Damgard}}, \bibinfo {author} {\bibfnamefont {S.}~\bibnamefont {Fehr}},
  \bibinfo {author} {\bibfnamefont {L.}~\bibnamefont {Salvail}}, \ and\
  \bibinfo {author} {\bibfnamefont {C.}~\bibnamefont {Schaffner}},\ }in\
  \href@noop {} {\emph {\bibinfo {booktitle} {Proceedings of the 46th Annual
  Symposium on Foundations of Computer Science}}}\ (\bibinfo {year} {2005})\
  pp.\ \bibinfo {pages} {449--458}\BibitemShut {NoStop}%
\bibitem [{\citenamefont {Hardy}\ and\ \citenamefont
  {Kent}(2004)}]{Kent:SensiCommit:PRL04}%
  \BibitemOpen
  \bibfield  {author} {\bibinfo {author} {\bibfnamefont {L.}~\bibnamefont
  {Hardy}}\ and\ \bibinfo {author} {\bibfnamefont {A.}~\bibnamefont {Kent}},\
  }\href@noop {} {\bibfield  {journal} {\bibinfo  {journal} {Physical Review
  Letters}\ }\textbf {\bibinfo {volume} {92}},\ \bibinfo {pages} {article no.
  157901} (\bibinfo {year} {2004})}\BibitemShut {NoStop}%
\bibitem [{\citenamefont {Wehner}\ \emph {et~al.}(2010)\citenamefont {Wehner},
  \citenamefont {Curty}, \citenamefont {Schaffner},\ and\ \citenamefont
  {Lo}}]{Wehner:noisystorage:PRA10}%
  \BibitemOpen
  \bibfield  {author} {\bibinfo {author} {\bibfnamefont {S.}~\bibnamefont
  {Wehner}}, \bibinfo {author} {\bibfnamefont {M.}~\bibnamefont {Curty}},
  \bibinfo {author} {\bibfnamefont {C.}~\bibnamefont {Schaffner}}, \ and\
  \bibinfo {author} {\bibfnamefont {H.-K.}\ \bibnamefont {Lo}},\ }\href@noop {}
  {\bibfield  {journal} {\bibinfo  {journal} {Physical Review A}\ }\textbf
  {\bibinfo {volume} {81}},\ \bibinfo {pages} {article no. 052336} (\bibinfo
  {year} {2010})}\BibitemShut {NoStop}%
\bibitem [{\citenamefont {Choi}\ \emph {et~al.}(2009)\citenamefont {Choi},
  \citenamefont {Hong}, \citenamefont {Chang}, \citenamefont {Chi},\ and\
  \citenamefont {Lee}}]{Choi:NonStatic:arXiv09}%
  \BibitemOpen
  \bibfield  {author} {\bibinfo {author} {\bibfnamefont {J.~W.}\ \bibnamefont
  {Choi}}, \bibinfo {author} {\bibfnamefont {D.}~\bibnamefont {Hong}}, \bibinfo
  {author} {\bibfnamefont {K.-Y.}\ \bibnamefont {Chang}}, \bibinfo {author}
  {\bibfnamefont {D.~P.}\ \bibnamefont {Chi}}, \ and\ \bibinfo {author}
  {\bibfnamefont {S.}~\bibnamefont {Lee}},\ }in\ \href@noop {} {\emph {\bibinfo
  {booktitle} {Proceedings of the 9th Asian Conference on Quantum Information
  Science}}}\ (\bibinfo {year} {2009})\ pp.\ \bibinfo {pages} {205--206},\
  \bibinfo {note} {also available at: arXiv:0901.1178v4}\BibitemShut {NoStop}%
\bibitem [{\citenamefont {D'Ariano}\ \emph {et~al.}(2007)\citenamefont
  {D'Ariano}, \citenamefont {Kretschmann}, \citenamefont {Schlingemann},\ and\
  \citenamefont {Werner}}]{DAriano:bitcommitment:PRA07}%
  \BibitemOpen
  \bibfield  {author} {\bibinfo {author} {\bibfnamefont {G.~M.}\ \bibnamefont
  {D'Ariano}}, \bibinfo {author} {\bibfnamefont {D.}~\bibnamefont
  {Kretschmann}}, \bibinfo {author} {\bibfnamefont {D.}~\bibnamefont
  {Schlingemann}}, \ and\ \bibinfo {author} {\bibfnamefont {R.~F.}\
  \bibnamefont {Werner}},\ }\href@noop {} {\bibfield  {journal} {\bibinfo
  {journal} {Physical Review A}\ }\textbf {\bibinfo {volume} {76}},\ \bibinfo
  {pages} {article no. 032328} (\bibinfo {year} {2007})}\BibitemShut {NoStop}%
\bibitem [{\citenamefont {Gisin}(1989)}]{Gisin:Stochastic:HPA89}%
  \BibitemOpen
  \bibfield  {author} {\bibinfo {author} {\bibfnamefont {N.}~\bibnamefont
  {Gisin}},\ }\href@noop {} {\bibfield  {journal} {\bibinfo  {journal}
  {Helvetica Physica Acta}\ }\textbf {\bibinfo {volume} {62}},\ \bibinfo
  {pages} {363} (\bibinfo {year} {1989})}\BibitemShut {NoStop}%
\bibitem [{\citenamefont {Hughstona}\ \emph {et~al.}(1993)\citenamefont
  {Hughstona}, \citenamefont {Jozsa},\ and\ \citenamefont
  {Wootters}}]{Hughstona:classification:PLA93}%
  \BibitemOpen
  \bibfield  {author} {\bibinfo {author} {\bibfnamefont {L.~P.}\ \bibnamefont
  {Hughstona}}, \bibinfo {author} {\bibfnamefont {R.}~\bibnamefont {Jozsa}}, \
  and\ \bibinfo {author} {\bibfnamefont {W.~K.}\ \bibnamefont {Wootters}},\
  }\href@noop {} {\bibfield  {journal} {\bibinfo  {journal} {Physics Letters
  A}\ }\textbf {\bibinfo {volume} {183}},\ \bibinfo {pages} {14} (\bibinfo
  {year} {1993})}\BibitemShut {NoStop}%
\bibitem [{\citenamefont {Jozsa}(1994)}]{Jozsa:MixedQuantum:JMO94}%
  \BibitemOpen
  \bibfield  {author} {\bibinfo {author} {\bibfnamefont {R.}~\bibnamefont
  {Jozsa}},\ }\href@noop {} {\bibfield  {journal} {\bibinfo  {journal} {Journal
  of Modern Optics}\ }\textbf {\bibinfo {volume} {41}},\ \bibinfo {pages}
  {2315} (\bibinfo {year} {1994})}\BibitemShut {NoStop}%
\end{thebibliography}%
\end{document}